\documentclass[12pt, reqno]{amsart}
\usepackage{amsmath, amsthm, amscd, amsfonts, amssymb, graphicx, color, enumitem, dsfont}
\usepackage[bookmarksnumbered, colorlinks, plainpages]{hyperref}
\hypersetup{colorlinks=true,linkcolor=red, anchorcolor=green, citecolor=cyan, urlcolor=red, filecolor=magenta, pdftoolbar=true}

\newcommand* {\ee}{\ensuremath{\mathrm{e}}}

\newtheorem{theorem}{Theorem}[section]
\newtheorem{lemma}[theorem]{Lemma}

\newtheorem{corollary}[theorem]{Corollary}
\theoremstyle{definition}

\newtheorem{example}[theorem]{Example}

\newtheorem*{remark}{Remark}
\numberwithin{equation}{section}

\begin{document}

\setcounter{page}{1}

\title[Product formulas in the framework of mean ergodic theorems]{Product formulas in the framework of mean ergodic theorems}

\author[J. Z. Bern\'ad]{J. Z. Bern\'ad $^1$}

\address{$^{1}$Department of Physics, University of Malta, Msida MSD 2080, Malta}
\email{\textcolor[rgb]{0.00,0.00,0.84}{zsolt.bernad@um.edu.mt}}

\let\thefootnote\relax\footnote{Copyright 2016 by the Tusi Mathematical Research Group.}



\date{Received: Mar. 20, 2019; Accepted: zzzzzz.}

\begin{abstract}
An extension of Chernoff's product formula for one-parameter functions taking values in the space of bounded linear operators on a Banach space is given. Essentially, the $n$-th
one-parameter function in the product formula is mapped by the $n$-th iterate of a contraction acting on the space of the one-parameter functions. The motivation to study
this specific product formula lies in the growing field of dynamical control of quantum systems, involving the procedure of dynamical decoupling and also the Quantum Zeno effect.
\end{abstract} \maketitle

\section{Introduction and preliminaries}
It has been shown recently the convergence in norm of the following product formula \cite{Bernad}
 \begin{equation}
  \lim_{n \to \infty} u \ee^{\frac{t}{n} x} u \ee^{\frac{t}{n} x} \dots u \ee^{\frac{t}{n} x} \left(u^\dagger\right)^n=\ee^{\mathcal{P}x\,t}.
\label{limit1}
 \end{equation}
with $u$ unitary operator and $x$ bounded operator on a Hilbert space $\mathcal{H}$. Furthermore, $t \in \mathbb{C}$ with $|t|<\infty$ and
$\mathcal{P}$ is a projection operator which maps the elements of the set
\begin{equation}
 \Xi_{\mathcal{B}(\mathcal{H})}=\{x \in \mathcal{B}(\mathcal{H}): \lim \frac{1}{n}\sum^n_{k=1} u^k x \left(u^\dagger\right)^k \,\, 
 \text{exists}\} \nonumber
\end{equation}
onto the linear subspace 
\begin{equation}
\{x \in \mathcal{B}(\mathcal{H}): [u,x]=0\}. \label{invset}
\end{equation}
The convergence of \eqref{limit1} is due to the properties of the power bounded map $x \rightarrow u x u^\dagger$, which has been discussed 
in the context of mean ergodic theorems \cite{Krengel}. 

Eq. \eqref{limit1} is an example of a degenerate semigroup product formula \cite{Arendt}. In these cases on may apply 
\cite{Matolcsi, Hillier} Chernoff's product formula \cite{Chernoff1, Chernoff2} rather than Trotter's \cite{Trotter1}. However, this approach
works for \eqref{limit1}, when there exists a non-zero natural number $k$ such that 
$u^k$ is equal to the identity operator $I$. We may even generalize this special case of \eqref{limit1} to the following product formula
\begin{equation}
\lim_{n \to \infty} \left(u_1 \ee^{\frac{t}{n} x} u_2 \ee^{\frac{t}{n} x} \dots u_k \ee^{\frac{t}{n} x} \right)^{\frac{n}{k}} \label{13}
\end{equation}
with $n/k$ being a natural number and $u_1u_2\dots u_k$ equals to $I$. In order to lift these type of constraints on the unitary operators, one has to provide a generalization of Chernoff's result.
A partial answer to this question has been given in Ref. \cite{Bernad}, where only uniformly continuous semigroups are considered. However, Chernoff's Theorem has a much deeper approach, and therefore
we intend to further develop the statements of \cite{Bernad} such that an extended version of Chernoff’s product formula is provided. This is the main aim of our paper, where we employ results from ergodic theory. 
It is also natural to consider the extension of the map  $x \rightarrow u x u^\dagger$ to a more general 
contraction map. This generalization will include automatically the cases of Refs. \cite{Matolcsi, Exner} where the unitary operator $u$ is replaced by a projection operator leading to questions 
in the topic of Quantum Zeno effect. In fact, there is a great interest from the physics community to understand the above mentioned product formulas \cite{Facchi} and approaches focusing mainly on
uniformly continuous semigroups are generating up to date results \cite{Barankai, Mobus}.

The approach to the main result in Theorem \ref{maintheorem}, presented in the subsequent discussion, is formulated in the language of strongly continuous semigroups and mean ergodic theorems. 
Therefore, it is assumed familiarity with the basic facts concerning these topics \cite{Krengel, Engel}. Recently, extensions of the Chernoff's product formula has been proved \cite{Vuillermot} 
in the context of non-autonomous differential equations \cite{Tanabe} and general Trotter-Kato formulas \cite{Vuillermot2}. In this article, we focus on a different direction, 
motivated by product formulas like \eqref{limit1}, and for this purpose we shall need the following mean ergodic theorem.

Let us consider a Banach space
$\mathcal{B}$ and a power bounded linear transformation $\Pi: \mathcal{B} \rightarrow \mathcal{B}$, i.e., there exists a
constant $M\geqslant1$ such that 
\begin{equation}
 \|\Pi^n\|\leqslant M \quad n \in \mathbb{N} \nonumber.
\end{equation}
$\Pi$ is a contraction when $M=1$.  We consider the following two closed linear subspaces for $\Pi$ \cite{Krengel}
\begin{equation}
 \Xi_{\mathcal{B}}=\{x \in \mathcal{B}: \lim \frac{1}{n}\sum^n_{k=1} \Pi^kx \,\, \text{exists}\}, \nonumber
\end{equation}
and
\begin{equation}
\text{Ker}(I-\Pi)=\{x \in \mathcal{B}: \Pi x=x\}. \nonumber
\end{equation}
When $\mathcal{B}$ is reflexive, the theorems of E. R. Lorch in \cite{Lorch} guarantee $\Xi_{\mathcal{B}}=\mathcal{B}$ and $\Pi$ is mean ergodic 
on $\mathcal{B}$. The next theorem on the separation of $\Xi_{\mathcal{B}}$ is mostly due to K. Yosida \cite{Krengel,Yosida}:
\begin{theorem}
 \label{theorem1}
 Let $\Pi$ be a power bounded linear operator on a Banach space $\mathcal{B}$. Then 
 \begin{equation}
  \Xi_{\mathcal{B}}=\text{Ker}(I-\Pi) \oplus \overline{\text{Rng}(I-\Pi)}. \nonumber
 \end{equation}
The linear operator $\mathcal{P}x=\lim \frac{1}{n}\sum^n_{k=1} \Pi^kx$ assigned to $x \in \Xi_{\mathcal{B}}$ is the projection of $\Xi_{\mathcal{B}}$ 
onto $\text{Ker}(I-\Pi)$. We have $\mathcal{P}=\mathcal{P}^2=\Pi\mathcal{P}=\mathcal{P}\Pi$ and for any $y\in \mathcal{B}$ the assertions
\begin{enumerate}[label=\alph*)]
 \item $\lim \frac{1}{n}\sum^n_{k=1} \Pi^ky=0$
 \item $y \in \overline{\text{Rng}(I-\Pi)}$
\end{enumerate}
are equivalent.
\end{theorem}

Recall also that a strongly continuous ($C_0$) semigroup $t \rightarrow T_t$ on $\mathcal{B}$ is such that
\begin{eqnarray}
 &1)& \quad T_0=\mathds{1}, \nonumber \\
 &2)& \quad T_sT_t=T_{t+s} \quad \forall t,s \geqslant 0, \nonumber \\
 &3)& \quad \lim_{t \to 0}\|T_t x - x\|=0 \quad \forall x \in \mathcal{B}. \nonumber
\end{eqnarray}
The strong derivative of $T_t$ at $t=0$ is a closed, densely-defined operator, the generator of $T_t$. From now on, we denote by $\mathcal{L}(\mathcal{B})$ the Banach space 
of all bounded linear operators on $\mathcal{B}$ endowed with the operator norm $\|.\|_{\mathcal{L}(\mathcal{B})}$.

\section{Main result}

Under the above conditions we are going to give an analytical proof of the following theorem:
\begin{theorem}\label{maintheorem}
 Let $V_t$ be a strongly continuous function from $\mathbb{R_+}$ to the linear contractions on the Banach space $\mathcal{B}$
 such that $V_0=\mathds{1}$. Consider a contraction $\Pi$ on $\mathcal{L}(\mathcal{B})$ such that $\Pi\mathds{1}=\mathds{1}$, whose Ces\`{a}ro mean
\begin{equation}
\frac{1}{n} \sum^n_{i=1} \Pi^ix, \quad x \in \mathcal{L}(\mathcal{B}) \nonumber  
\end{equation}
converges to the projector $\mathcal{P}$ projecting onto the linear subspace
 \begin{equation}
\text{Ker}(I-\Pi)=\{x \in \mathcal{L}(\mathcal{B}): \Pi x=x\}. \nonumber
\end{equation}
Suppose that $V_t \in \Xi_{\mathcal{L}(\mathcal{B})}$ and the closure $\overline{A}$ of $A=\frac{d}{dt} \mathcal{P}(V_t)$ at $t=0$ is the generator of a $C_0$ 
contraction semigroup. Then
\begin{equation}
 \Pi(V_{t/n})\Pi^2(V_{t/n}) \dots \Pi^n(V_{t/n}) \nonumber
\end{equation}
converges to $e^{t\overline{A}}$ in the strong operator topology.
\end{theorem}
The theorem will be proved via a theorem and three lemmas, whereas some of them are already known.
\begin{theorem} \label{lemma1}
Let $\{T_n(t)\}_{t \geqslant 0}$, $n=0,1,2,\dots$, be $C_0$ semigroups on the Banach space $\mathcal{B}$ with generators $A_n$ satisfying 
the stability condition
\begin{equation}
 \|T_n(t) \|\leqslant M e^{\omega t} \label{stability}
\end{equation}
where the constants $M \geqslant 1$ and $\omega \in \mathbb{R}$ are independent of $n$ and $t$. Define $A=\lim_{n \to \infty} A_n$ and suppose that
\begin{enumerate}[label=\alph*)]
 \item $A$ is densely defined,
 \item for some $\lambda > \omega$ the range $\text{Rng}(\lambda-A)$ is dense in $\mathcal{B}$.
\end{enumerate}
Then the closure of $A$ is the infinitesimal generator of a $C_0$ semigroup $T(t)$, which satisfies $T(t)=\lim_{n \to \infty} T_n(t)$ in the strong operator topology
and uniformly for every compact interval $[0, t_0]$.
\end{theorem}
\begin{proof}
 It is in fact the Trotter-Kato approximation theorem, see \cite{Trotter2} together with \cite{Kato}. Another approach in 
 \cite{Chernoff1} considers $\mathcal{B}$ to be separable.
\end{proof}

\begin{lemma} \label{lemma2}
Let $S$ be a linear contraction on the Banach space $\mathcal{B}$. Then we have
\begin{equation}
\|e^{n(S-I)}x-S^nx \| \leqslant \sqrt{n} \|Sx-x\|
\end{equation}
for all $x \in \mathcal{B}$ and $n=0,1,2,\dots$
\end{lemma}
\begin{proof}
 Can be found in Lemma $2$ of \cite{Chernoff1}.
\end{proof}
\begin{lemma}\label{lemma3}
 A sequence of operators $T_n \in \mathcal{L}(\mathcal{B})$, where $\mathcal{B}$ is a Banach space, converges strongly to an operator $T \in \mathcal{L}(\mathcal{B})$ iff
 \begin{enumerate}[label=\alph*)]
 \item the sequence $\{T_n x\}$ converges for any $x$ in a dense subset $D\subset\mathcal{B}$
 \item and there exists $M>0$ such that $\|T_n\|_{\mathcal{L}(\mathcal{B})} \leqslant M$.
\end{enumerate}
\end{lemma}
\begin{proof}
Is a special case of the proof of Theorem II.$1.18$ in \cite{Dunford}.  
\end{proof}
\begin{lemma} \label{lemma4}
Let $\Pi$ be a contraction on a Banach space $\mathcal{B}$ such that $\Pi\mathds{1}=\mathds{1}$. Suppose that $x \in 
\Xi_{\mathcal{B}}$ with $\| \mathds{1}+\frac{t}{n} x\|\leqslant1$ for all $t > 0$ and $n=1,2,\dots$. Then, together with the 
projector $\mathcal{P}$ projecting onto the linear subspace $\text{Ker}(I-\Pi)$
\begin{eqnarray}
 \lim_{n \to \infty}\| (\mathds{1}+\frac{t}{n} \mathcal{P}x )^n-(\mathds{1}+\frac{t}{n} \Pi x) (\mathds{1}+\frac{t}{n} \Pi^2x) \dots (\mathds{1}+\frac{t}{n} \Pi^nx) \|=0. \label{toprove}
\end{eqnarray}
\end{lemma}
\begin{proof}
 Let us fix $n$. According to Theorem \ref{theorem1}
 \begin{equation}
  \Xi_{\mathcal{B}}=\text{Ker}(I-\Pi) \oplus \overline{\text{Rng}(I-\Pi)}. \nonumber
 \end{equation}
Suppose first that $x \in \text{Ker}(I-\Pi) \oplus \text{Rng}(I-\Pi)$, so
therefore exists a $y \in \mathcal{B}$ such that $x=\mathcal{P}x+y-\Pi y$ and we define the following operator:
\begin{equation}
 S_i= (\mathds{1}+t \mathcal{P}x)^i  (\mathds{1}+t \Pi y-t \Pi^{i+1}y) (\mathds{1}+t \Pi^{i+1}x) \dots (\mathds{1}+t \Pi^{n}x),\, 1\leqslant i < n. \nonumber
\end{equation} 
Then
\begin{eqnarray}
 &&\| (\mathds{1}+t \mathcal{P}x )^n-(\mathds{1}+t \Pi x) (\mathds{1}+t \Pi^2x) \dots (\mathds{1}+t \Pi^nx) \| \nonumber \\
&& \leqslant \| (\mathds{1}+t \mathcal{P}x)^n -(\mathds{1}+t \mathcal{P}x)^n  (\mathds{1}+t \Pi y-t \Pi^{n+1}y)\| \nonumber \\
&&+\|(\mathds{1}+t \mathcal{P}x)^n  (\mathds{1}+t \Pi y-t \Pi^{n+1}y)
 -S_{n-1}\|+\sum^{n-2}_{i=1} \|S_{i+1} -S_{i}\| \nonumber \\
 &&+\| S_1-(\mathds{1}+t \Pi x) (\mathds{1}+t \Pi^2x) \dots (\mathds{1}+t \Pi^nx) \|. \label{ineqtodo}
\end{eqnarray}
We are going to upper bound every term on the right hand side of the above inequality. For the first term, we have 
\begin{eqnarray}
 &&\| (\mathds{1}+t \mathcal{P}x)^n -(\mathds{1}+t \mathcal{P}x)^n  (\mathds{1}+t \Pi y-t \Pi^{n+1}y)\| \nonumber \\
 &&\leqslant \|\mathcal{P} (\mathds{1} + t x) \|^n \| t \Pi y-t \Pi^{n+1}y\|_{\mathcal{B}} \leqslant 2 t \|y \|, \label{ineq1}
\end{eqnarray}
where we have used that $\mathcal{P}$ is a projector, $\Pi$ a contraction and $\| \mathds{1}+t x\|\leqslant1$. The other terms yield
\begin{eqnarray}
 \|(\mathds{1}+t \mathcal{P}x)^n  (\mathds{1}+t \Pi y-t \Pi^{n+1}y)-S_{n-1}\| \leqslant 4 t^2 \left(\|x \| \|y \|+\|y \|^2 \right), \label{ineq2} 
\end{eqnarray}
for $1\leqslant i < n-1$
\begin{eqnarray}
 \|S_{i+1} -S_{i}\| &\leqslant& \|\mathcal{P} (\mathds{1} + t x) \|^i 4 t^2 \left(\|x \| \|y \|+\|y \|^2 \right) \underbrace{\|\Pi^{i+2} (\mathds{1} + t x) \|\dots}_{n-i-1 \, \text{terms}} \nonumber \\
 &\leqslant& 4 t^2 \left(\|x \| \|y \|+\|y \|^2 \right), \label{ineq3}
\end{eqnarray}
and finally
\begin{eqnarray}
 \| S_1-(\mathds{1}+t \Pi x) (\mathds{1}+t \Pi^2x) \dots (\mathds{1}+t \Pi^nx) \| \leqslant 4 t^2 \left(\|x \| \|y \|+\|y \|^2 \right). \nonumber \\ \label{ineq4}
\end{eqnarray}
Substituting inequalities \eqref{ineq1}, \eqref{ineq2}, \eqref{ineq3}, and \eqref{ineq4} into \eqref{ineqtodo}
results that for an element $x \in \text{Ker}(I-\Pi) \oplus \text{Rng}(I-\Pi)$ with $\| \mathds{1}+t x\|\leqslant1$ for all $t \geqslant 0$
\begin{eqnarray}
 &&\| (\mathds{1}+t \mathcal{P}x )^n-(\mathds{1}+t \Pi x) (\mathds{1}+t \Pi^2x) \dots (\mathds{1}+t \Pi^nx) \| \nonumber \\
&& \leqslant 2 t \|y \| + 4 n t^2 \left(\|x \| \|y \|+\|y \|^2 \right). \nonumber
\end{eqnarray}
From the substitution $t\rightarrow t/n $, we obtain
\begin{eqnarray}
 &&\| (\mathds{1}+\frac{t}{n} \mathcal{P}x )^n-(\mathds{1}+\frac{t}{n} \Pi x) (\mathds{1}+\frac{t}{n} \Pi^2x) \dots (\mathds{1}+\frac{t}{n} \Pi^nx) \| \nonumber \\
 && \leqslant 2 t \frac{\|y \|}{n} + 4 t^2 \frac{\|x \| \|y \|+\|y \|^2}{n}. \label{touse}
\end{eqnarray}
Thus \eqref{toprove} holds for an element $x \in \text{Ker}(I-\Pi) \oplus \text{Rng}(I-\Pi)$. We claim that \eqref{toprove} holds also for an element 
$x \in \text{Ker}(I-\Pi) \oplus \overline{\text{Rng}(I-\Pi)}$. Let $x_1,x_2 \in \Xi_{\mathcal{B}}$ such that $\| \mathds{1}+\frac{t}{n} x_i\|\leqslant1$ ($i=1,2$) for all $t \geqslant 0$ and $n=1,2,\dots$.
Then,
\begin{eqnarray}
 &&\| (\mathds{1}+\frac{t}{n} \Pi x_1) \dots (\mathds{1}+\frac{t}{n} \Pi^nx_1)-(\mathds{1}+\frac{t}{n} \Pi x_2) \dots (\mathds{1}+\frac{t}{n} \Pi^nx_2) \| \nonumber \\
 &&\leqslant \| (\mathds{1}+\frac{t}{n} \Pi x_1) \dots (\mathds{1}+\frac{t}{n} \Pi^nx_1)-T_{n-1}\|+\sum^{n-2}_{i=1} \|T_{i+1}-T_{i}\|  \nonumber \\
  &&+ \|T_1-(\mathds{1}+\frac{t}{n} \Pi x_2) \dots (\mathds{1}+\frac{t}{n} \Pi^nx_2) \| \nonumber
\end{eqnarray}
with
\begin{equation}
T_i=(\mathds{1}+\frac{t}{n} \Pi x_1) \dots (\mathds{1}+\frac{t}{n} \Pi^ix_1)(\mathds{1}+\frac{t}{n} \Pi^{i+1}x_2)\dots (\mathds{1}+\frac{t}{n} \Pi^nx_2) \nonumber 
\end{equation}
yields
\begin{eqnarray}
 \| (\mathds{1}+\frac{t}{n} \Pi x_1) \dots (\mathds{1}+\frac{t}{n} \Pi^nx_1)-(\mathds{1}+\frac{t}{n} \Pi x_2) \dots (\mathds{1}+\frac{t}{n} \Pi^nx_2) \| \leqslant t \|x_1-x_2\|. \nonumber \\
 \label{ineq5}
\end{eqnarray}
Let $x_2=\mathcal{P}x_2+z$, where $z$ is not of the form $y-\Pi y$. Then for any $\epsilon > 0$ we can take 
an arbitrary $\epsilon' > 0$ satisfying
\begin{equation}
  \epsilon'(t+1) \leqslant \epsilon \nonumber
\end{equation}
and for this $\epsilon'$ there exists $x_1 \in \text{Ker}(I-\Pi) \oplus \text{Rng}(I-\Pi)$ such that
\begin{equation}
 \mathcal{P}x_1 =\mathcal{P}x_2, \quad \text{and} \quad \|x_1-x_2\|< \epsilon'. \nonumber
\end{equation} 
Furthermore, for $\epsilon'> 0$ we can take an $n_0$ such that for every $n>n_0$
\begin{eqnarray}
 \| (\mathds{1}+\frac{t}{n} \mathcal{P}x_2 )^n-(\mathds{1}+\frac{t}{n} \Pi x_1) (\mathds{1}+\frac{t}{n} \Pi^2x_1) \dots (\mathds{1}+\frac{t}{n} \Pi^nx_1) \| < \epsilon', \nonumber \\
 \label{ineq6}
\end{eqnarray}
which is due to the inequality in \eqref{touse}. Thus, with the combination of \eqref{ineq5} and \eqref{ineq6} we have for every $n>n_0$
\begin{eqnarray}
 &&\| (\mathds{1}+\frac{t}{n} \mathcal{P}x_2 )^n-(\mathds{1}+\frac{t}{n} \Pi x_2) (\mathds{1}+\frac{t}{n} \Pi^2x_2) \dots (\mathds{1}+\frac{t}{n} \Pi^nx_2) \| \nonumber \\
 && \leqslant \| (\mathds{1}+\frac{t}{n} \mathcal{P}x_2 )^n-(\mathds{1}+\frac{t}{n} \Pi x_1) (\mathds{1}+\frac{t}{n} \Pi^2x_1) \dots (\mathds{1}+\frac{t}{n} \Pi^nx_1) \| \nonumber \\
 &&+\| (\mathds{1}+\frac{t}{n} \Pi x_1) \dots (\mathds{1}+\frac{t}{n} \Pi^nx_1)-(\mathds{1}+\frac{t}{n} \Pi x_2) \dots (\mathds{1}+\frac{t}{n} \Pi^nx_2) \| \nonumber \\
 && < \epsilon'+ t \|x_1-x_2\|=\epsilon'(t+1) \leqslant \epsilon.
\end{eqnarray}
As $\epsilon >0$  was arbitrary, we have proved our statement.
\end{proof}

\begin{proof}[Proof of Theorem \ref{maintheorem}]
 Fix $t>0$ and define
 \begin{equation}
  A_n=\frac{V_{t/n}-\mathds{1}}{t/n} \in \mathcal{L}(\mathcal{B}). \nonumber
 \end{equation}
The domain $D$ of the strong derivative $\| (V_t-\mathds{1})/t \,x\|$ consists of those $x$ for which the limit $t \to 0$ exists and its a dense subspace of $\mathcal{B}$. Replacing $V_t$ with $\mathcal{P}(V_t)$
it is unequivocal that the domain of those $x$ for which $\lim_{t \to 0}\| \left(\mathcal{P}(V_t)-\mathds{1}\right)/t \,x\|$ ($\mathcal{P}\mathds{1}=\mathds{1} $) exists it coincides with $D$.
The semigroups $e^{t \mathcal{P} (A_n)}$ all satisfy 
\begin{eqnarray}
\|e^{t \mathcal{P} (A_n)} \|_{\mathcal{L}(\mathcal{B})} &\leqslant& e^{-n} \|e^{n \mathcal{P} (V_{t/n})} \|_{\mathcal{L}(\mathcal{B})} \leqslant 
e^{-n} \sum^\infty_{i=0} \frac{n^i  \|\mathcal{P} (V_{t/n}) \|^i_{\mathcal{L}(\mathcal{B})}}{i!} \nonumber \\
& \leqslant& e^{-n} \sum^\infty_{i=0} \frac{n^i}{i!}=1. \label{stabproved}
\end{eqnarray}
Thus, the stability condition \eqref{stability} of Theorem \ref{lemma1} is fulfilled . We define 
\begin{equation}
 A=\lim_{t \to 0} \frac{\mathcal{P}(V_t)-\mathds{1}}{t},
\end{equation}
and it is obvious that $\lim_{n \to \infty} \mathcal{P} (A_n)x=Ax$ for all $x \in D$. As the closure $\overline{A}$ of $A$ generates a $C_0$ semigroup $e^{t \overline{A}}$, then together
with the stability condition these arguments yield that $e^{t \overline{A}}$ is a contraction semigroup, which satisfies
\begin{equation}
 \|e^{t \overline{A}}x-e^{t \mathcal{P} (A_n)}x\| \to 0 \label{p1}
\end{equation}
for all $x \in \mathcal{B}$ as $n \to \infty$ and uniformly for $t\in[0,t_0]$.

On the other hand
\begin{eqnarray}
 \|e^{t \mathcal{P} (A_n)}x- \left[\mathcal{P}(V_{t/n})\right]^n x\|=\|e^{n \mathcal{P} (V_{t/n})-n}x- \left[\mathcal{P}(V_{t/n})\right]^n x\|, \nonumber 
\end{eqnarray}
which with the notation $S=\mathcal{P}(V_{t/n})$ and Lemma \ref{lemma2} yields
\begin{eqnarray}
 \|e^{t \mathcal{P} (A_n)}x- \left[\mathcal{P}(V_{t/n})\right]^n x\|&\leqslant& \sqrt{n} \|\mathcal{P}(V_{t/n})x-x \| \nonumber \\
 &\leqslant&\frac{t}{\sqrt{n}} \|\mathcal{P} (A_n) x \| \to 0 \nonumber
\end{eqnarray}
for all $x \in D$ as $n \to \infty$ and uniformly for $t\in (0,t_0]$.
\begin{equation}
 \|e^{t \mathcal{P} (A_n)}- \left[\mathcal{P}(V_{t/n})\right]^n \|_{\mathcal{L}(\mathcal{B})} \leqslant 2 \nonumber
\end{equation}
is due to \eqref{stabproved} and using Lemma \ref{lemma3}, we obtain
\begin{equation}
 \|e^{t \mathcal{P} (A_n)}x- \left[\mathcal{P}(V_{t/n})\right]^n x\| \to 0 \label{p2}
\end{equation}
for all $x \in \mathcal{B}$ as $n \to \infty$.

Lastly, we have
\begin{eqnarray}
 \Pi(V_{t/n})\Pi^2(V_{t/n})\dots \Pi^n(V_{t/n})=(\mathds{1}+\frac{t}{n} \Pi A_n) (\mathds{1}+\frac{t}{n} \Pi^2 A_n) \dots (\mathds{1}+\frac{t}{n} \Pi^nA_n). \nonumber
\end{eqnarray}
As
\begin{equation}
 \|\Pi(V_{t/n})\Pi^2(V_{t/n})\dots \Pi^n(V_{t/n})-\left[\mathcal{P}(V_{t/n})\right]^n\|_{\mathcal{L}(\mathcal{B})} \leqslant 2 \nonumber
\end{equation}
for all $t \geqslant 0$ and $n=1,2,\dots$, the application of Lemma \ref{lemma4} first and then Lemma \ref{lemma3} results
\begin{eqnarray}
 \|\Pi(V_{t/n})\Pi^2(V_{t/n})\dots \Pi^n(V_{t/n})x- \left[\mathcal{P}(V_{t/n})\right]^nx\|\to 0 \label{p3}
\end{eqnarray}
for all $x \in \mathcal{B}$ as $n \to \infty$. A combination of \eqref{p1}, \eqref{p2} and \eqref{p3} proves the statement of the theorem.
\end{proof}

\begin{example}Let us consider the Hilbert space $L^2(\mathbb{R})$ of the square-integrable functions on the real line. The Fourier transform 
\begin{equation}
 (\mathcal{F}f)(x)=\lim_{N \to \infty} \frac{1}{\sqrt{2 \pi}} \int_{|y| \leqslant N} e^{-ixy} f(y) dy \nonumber
\end{equation}
is a unitary map of $L^2(\mathbb{R})$ onto $L^2(\mathbb{R})$ with the property
\begin{equation}
 (\mathcal{F}^{-1}f)(x)=(\mathcal{F}f)(-x), \quad \forall f \in L^2(\mathbb{R}). \nonumber
\end{equation}
Define the following $C_0$ semigroups
\begin{eqnarray}
 (V_t f)(x)&=&e^{-itx}f(x), \nonumber \\
 (U_t f)(x)&=&f(x-t), \quad \forall f \in L^2(\mathbb{R}), \nonumber
\end{eqnarray}
and
\begin{equation}
 \Pi (V_t)=\mathcal{F} V_t \mathcal{F}^{-1}. \nonumber
\end{equation}
According to Theorem \ref{maintheorem}, we have to determine the projector $\mathcal{P}$ which projects onto the linear subspace $\text{Ker}(I-\Pi)$. However,
\begin{equation}
 \mathcal{F} V_t \mathcal{F}^{-1}=U_{-t}, \quad \mathcal{F}^{-1} V_t \mathcal{F}=U_{t}, \nonumber
\end{equation}
which shows that $\mathcal{P}(V_t)=\mathds{1}$ (the identity operator on $L^2(\mathbb{R})$) and therefore
\begin{equation}
 \| \Pi(V_{t/n})\Pi^2(V_{t/n})\Pi^3(V_{t/n}) \dots \Pi^n(V_{t/n})f-f\|\to 0 \nonumber 
\end{equation}
for all $f \in L^2(\mathbb{R})$ as $n \to \infty$.
\end{example}
\begin{example}
Let $T_1(t)$ and $T_2(t)$ be two $C_0$ contraction semigroups on a Banach space $\mathcal{B}$. The generators are $A_1$ and $A_2$ with the respective 
dense domains $D(A_1)$ and $D(A_2)$. Consider the semigroup
\begin{equation}
 V_t=\begin{pmatrix}
     \frac{T_1(t)+T_2(t)}{2} & \frac{T_1(t)-T_2(t)}{2} \\
     \frac{T_1(t)-T_2(t)}{2} & \frac{T_1(t)+T_2(t)}{2}  
     \end{pmatrix} \nonumber
\end{equation}
on the direct sum $\mathcal{B} \oplus \mathcal{B}$, which is also a Banach space. Then, the contraction map 
\begin{equation}
 \Pi \begin{pmatrix}
     a & b \\
     c & d  
     \end{pmatrix}=\begin{pmatrix}
     a & -b \\
     -c & d  
     \end{pmatrix}, \quad a,b,c,d \in \mathcal{L}(\mathcal{B}) \nonumber
\end{equation}
implies that 
\begin{equation}
 \mathcal{P}(V_t)=\begin{pmatrix}
     \frac{T_1(t)+T_2(t)}{2} & 0 \\
     0 & \frac{T_1(t)+T_2(t)}{2}  
     \end{pmatrix}. \nonumber
\end{equation}
The generator $A$ in Theorem \ref{maintheorem} reads
\begin{equation}
 A=\frac{A_1+A_2}{2}\mathds{1},  \nonumber
\end{equation}
where $\mathds{1}$ is the identity operator on $\mathcal{B} \oplus \mathcal{B}$ and the domain of $(A_1+A_2)/2$ is $D(A_1)\cap D(A_2)$, a subspace that might be $\{0\}$ in general.
\end{example}
\begin{corollary}\label{cor}
Let $V: \mathbb{R_+} \to \mathcal{L}(\mathcal{B})$ and $\Pi$ a contraction on $\mathcal{L}(\mathcal{B})$ satisfying the assumptions of Theorem \ref{maintheorem}. If for 
fixed $t>0$ we take a a strictly increasing sequence of natural numbers $(k_n)_{n \in \mathbb{N}}$ and a positive null sequence $(t_n)_{n \in \mathbb{N}}$ such that $k_nt_n \to t$
as $n \to \infty$, then
\begin{equation}
 \Pi(V_{t_n})\Pi^2(V_{t_n}) \dots \Pi^{k_n}(V_{t_n}) \nonumber
\end{equation}
converges to $e^{t\overline{A}}$ in the strong operator topology.
\end{corollary}
\begin{proof}
 Applying Lemma \ref{lemma4}, we have
\begin{eqnarray}
\|\Pi(V_{t_n})\Pi^2(V_{t_n})\dots \Pi^{k_n}(V_{t_n})x - \left[\mathcal{P}(V_{t_n})\right]^{k_n}x \| \to 0\label{p4}
\end{eqnarray}
for all $x \in D$ as $n \to \infty$ and uniformly for $t\in (0,t_0]$. Since $V_{t/n}$s and $\Pi$ are contractions, Lemma \ref{lemma3} yields that we have the result for all vectors in $\mathcal{B}$.
The proof of
\begin{equation}
 \lim_{n \to \infty} \|e^{t \overline{A}}x-\left[\mathcal{P}(V_{t_n})\right]^{k_n}x\|= 0 \nonumber
\end{equation}
for all $x \in \mathcal{B}$ is given in Corollary $5.4$ of \cite{Engel}.
\end{proof}
\begin{example}
Let us consider the product formula in \eqref{13} with only two unitary operators
\begin{equation}
\lim_{n \to \infty} \left(u_1 \ee^{\frac{t}{n} x} u_2 \ee^{\frac{t}{n} x} \right)^{\frac{n}{2}}, \nonumber  
\end{equation}
where $x$ is the generator of the $C_0$ contraction semigroup $e^{tx}$ acting on a Hilbert space $\mathcal{H}$ and  $u_1u_2=u \neq I$. Then, 
\begin{eqnarray}
&&\lim_{n \to \infty} \left(u_1 \ee^{\frac{t}{n} x} u_2 \ee^{\frac{t}{n} x} \right)^{\frac{n}{2}} \left(u^\dagger\right)^{\frac{n}{2}}=\lim_{n \to \infty} V_{2t/n}\Pi(V_{2t/n}) 
\dots \Pi^{\frac{n}{2}-1}(V_{2t/n}), \nonumber \\
&&V_{t}=u_1 \ee^{\frac{t}{2} x} u^\dagger_1 u \ee^{\frac{t}{2} x} u^\dagger, \quad  \Pi x= u x u^\dagger. \nonumber
\end{eqnarray}
$\mathcal{P}$ projects onto the linear subspace 
\begin{equation}
\{x \in \mathcal{B}(\mathcal{H}): [u,x]=0\}. \nonumber
\end{equation}
When $x$ is bounded, Theorem \ref{maintheorem} and Corollary \ref{cor} yields
\begin{equation}
 \overline{A}=A=\frac{1}{2} \mathcal{P}\left( u_1 x u^\dagger_1 + u x u^\dagger\right), \nonumber
\end{equation}
otherwise
\begin{equation}
 A=\lim_{t \to 0} \frac{\mathcal{P}(V_t)-\mathds{1}}{t}. \nonumber
\end{equation}
\end{example}
\begin{remark}
We have provided an extended version of Chernoff's product formula with iterations of a contraction acting on the one-parameter family of linear contractions. The result can be applied to many
mathematical problems in the field of dynamical control of quantum systems. However, when the bare evolution of the system is described by a non-autonomous differential equation then the results of
Ref. \cite{Vuillermot} should be generalized in the direction discussed here. We leave these questions open for now.
\end{remark}
 
\section*{Acknowledgement}

This work is supported by the European Union's Horizon 2020 research and innovation programme under 
Grant Agreement No. 732894 (FET Proactive HOT).

\bibliographystyle{amsplain}

\end{document}